\documentclass[11pt]{article}
\usepackage{fullpage}
\usepackage{authblk}
\usepackage[linesnumbered,ruled,vlined]{algorithm2e}
\usepackage{amssymb}
\usepackage{amsmath}
\usepackage{amsthm}
\usepackage{mathtools}
\usepackage{appendix}
\usepackage{color}
\usepackage{comment}
\usepackage{subfig}
\usepackage{pstricks}
\usepackage{float}
\usepackage{graphics}
\usepackage[dvips,final]{graphicx}
\usepackage{epsfig}

\usepackage{algorithmic}

\newtheorem{lem}{Lemma}
\newtheorem{cor}{Corollary}
\newtheorem{thm}{Theorem}
\newtheorem{cl}{Claim}
\newtheorem{defn}{Definition}
\newtheorem{fact}{Fact}
\newtheorem{result}{Result}

   \newtheorem{innercustomthm}{Theorem}
 \newenvironment{customTheorem}[1]
   {\renewcommand\theinnercustomthm{#1}\innercustomthm}
   {\endinnercustomthm}

\begin{document}

\title{Efficiency of Truthful and Symmetric Mechanisms in One-sided Matching\thanks{This work was partially supported by ERC StG project PAAl 259515, FET IP project MULTIPEX 317532, and NCN grant N N206 567940.}}

\author[1]{Marek Adamczyk}
\author[2]{Piotr Sankowski}
\author[2]{Qiang Zhang}
\affil[1]{Sapienza University of Rome, Italy\\ \tt{adamczyk@dis.uniroma1.it}}
\affil[2]{Institute of Informatics, University of Warsaw\\ \tt{sank,qzhang@mimuw.edu.pl}}


\maketitle

\global\long\def\adj{\mbox{\footnotesize Adj}}

\global\long\def\br#1{\left( #1 \right)}

\global\long\def\brq#1{\left[ #1 \right]}

\global\long\def\brw#1{\left\{  #1\right\}  }

\global\long\def\cut#1{\partial#1 }

\global\long\def\br#1{\left( #1 \right)}

\global\long\def\brq#1{\left[ #1 \right]}

\global\long\def\chr#1{\mathbf{1}\brq{#1}}

\global\long\def\brw#1{\left\{  #1\right\}  }

\global\long\def\br#1{\left( #1 \right)}

\global\long\def\brq#1{\left[ #1 \right]}

\global\long\def\brw#1{\left\{  #1\right\}  }

\global\long\def\excond#1#2{\mathbb{E}\left[\left. #1 \right\vert #2 \right]}

\global\long\def\ex#1{\mathbb{E}\left[#1\right]}

\global\long\def\exls#1#2{\mathbb{E}_{#1}\left[#2\right]}

\global\long\def\pr#1{\mathbb{P}\left[ #1 \right]}

\global\long\def\prcond#1#2{\mathbb{P}\left[\left. #1 \right\vert #2 \right]}

\global\long\def\setst#1#2{\left\{  \left.#1\right|#2\right\}  }

\global\long\def\set#1{\left\{  #1\right\}  }

\global\long\def\ind#1{\mathbf{1}\left[ #1 \right]}

\global\long\def\st#1{[#1] }

\global\long\def\opstyle#1{\mathbb{#1}}

\global\long\def\ex#1{\mathbb{E}\left[#1\right]}

\global\long\def\xp#1{\mathbb{E}\left[#1\right]}

\global\long\def\excond#1#2{\opstyle E \left[\left. #1 \right\vert #2 \right]}

\global\long\def\exls#1#2{\opstyle{\opstyle E}_{#1}\left[ #2 \right]}

\global\long\def\size#1{\left|#1\right|}

\global\long\def\setst#1#2{\left\{  #1\left|#2\right.\right\}  }

\global\long\def\setst#1#2{\left\{  \left.#1\right|#2\right\}  }

\global\long\def\setstcol#1#2{\left\{  #1:#2\right\}  }

\global\long\def\set#1{\left\{  #1\right\}  }

\global\long\def\adj{\mbox{\footnotesize Adj}}

\global\long\def\st#1{[#1] }

\global\long\def\indi#1{\chi\brq{#1}}

\global\long\def\evalat#1#2{ #1 \Big|_{#2}}

\global\long\def\Df#1#2{\frac{\partial#1}{\partial#2}}

\global\long\def\ex#1{\mathbb{E}\left[#1\right]}

\global\long\def\xp#1{\mathbb{E}\left[#1\right]}

\global\long\def\prls#1#2{\opstyle P_{#1}\left[ #2 \right]}

\global\long\def\pr#1{\opstyle P \left[ #1 \right]}

\global\long\def\prcond#1#2{\opstyle P \left[\left. #1 \right\vert #2 \right]}

\global\long\def\pr#1{\opstyle P \left[ #1 \right]}

\global\long\def\prcond#1#2{\opstyle P \left[\left. #1 \right\vert #2 \right]}

\global\long\def\excond#1#2{\opstyle E \left[\left. #1 \right\vert #2 \right]}

\global\long\def\exls#1#2{\opstyle{\opstyle E}_{#1}\left[ #2 \right]}

\global\long\def\prls#1#2{\opstyle P_{#1}\left[ #2 \right]}

\global\long\def\size#1{\left|#1\right|}

\global\long\def\setst#1#2{\left\{  #1\left|#2\right.\right\}  }

\global\long\def\setst#1#2{\left\{  \left.#1\right|#2\right\}  }

\global\long\def\setstcol#1#2{\left\{  #1:#2\right\}  }

\global\long\def\set#1{\left\{  #1\right\}  }

\global\long\def\adj#1{\delta\br{#1}}

\global\long\def\st#1{[#1] }

\global\long\def\indi#1{\chi\brq{#1}}

\global\long\def\evalat#1#2{ #1 \Big|_{#2}}

\global\long\def\Df#1#2{\frac{\partial#1}{\partial#2}}

\global\long\def\eps{\varepsilon}

\global\long\def\rsd{\mathcal{R}}

\global\long\def\opt{\mathcal{O}}

\global\long\def\kvv{\mathcal{R}^*}

\global\long\def\profbrack#1{\left\langle{#1}\right\rangle}

\global\long\def\profi#1{\mathbb{I}_{#1}}

\global\long\def\vect#1{#1}

\global\long\def\probnamem{\mbox{one-sided matching }}

\global\long\def\probname{one-sided matching }

\global\long\def\sw#1{\nu\br{#1}}

\begin{abstract}
We study the efficiency (in terms of social welfare) of truthful and symmetric 
mechanisms in one-sided matching problems with {\em dichotomous preferences}
and {\em normalized von Neumann-Morgenstern preferences}. We are particularly
interested in the well-known {\em Random Serial Dictatorship} mechanism. For 
dichotomous preferences, we first show that truthful, symmetric and optimal 
mechanisms exist if intractable mechanisms are allowed. We then provide a 
connection to online bipartite matching. Using this connection, it is possible 
to design truthful, symmetric and tractable mechanisms that extract 
$0.69$ of the maximum social welfare, which works under assumption that
agents are not adversarial. Without this assumption, we show that Random Serial
Dictatorship always returns an assignment in which the expected social welfare
is at least a third of the maximum social welfare. For normalized von
Neumann-Morgenstern preferences, we show that Random Serial Dictatorship
always returns an assignment in which the expected social welfare is at
least $\frac{1}{e}\frac{\nu(\opt)^2}{n}$, where $\nu(\opt)$ is the maximum
social welfare and $n$ is the number of both agents and items. On the hardness
side, we show that no truthful mechanism can achieve a social welfare better
than $\frac{\nu(\opt)^2}{n}$.
\end{abstract}

\section{Introduction}

We study the efficiency of mechanisms in one-sided matching problems, where the
goal is to allocate $n$ {\em indivisible} items to $n$ unit-demand
{\em rational} agents having \textit{private} preferences over items. Agents
are rational, i.e., they would like to be assigned to the best items according
to their private preferences. The problem essentially captures variants of
practical applications such as allocating houses to residents, assigning
professors to courses and so on. In this paper, we mainly focus on 
{\em cardinal preferences} in which agents have values for different items.
A practical setting would be that residents have values for different houses.
A {\em mechanism} maps preferences that agents report to a matching, which
is a one-to-one mapping between agents and items. Throughout the paper,
depending on the context, we use sometimes term {\em matching} and sometimes
{\em assignment}, but they always mean essentially the same. One immediate
question arises: if there exist mechanisms in which no agent could benefit 
by deviating from reporting his true preference regardless the preferences 
reported by other agents? Such mechanisms are often called  {\em truthful}
mechanisms.
The question was answered in~\cite{svensson1999strategy}, where it was shown
that there exists only one {\em truthful, nonbossy} and {\em neutral}
mechanism. A mechanism is nonbossy if an individual agent cannot change 
the output of the mechanism without changing his assignment. A mechanism
is neutral if the mechanism is independent of the identities of items, e.g.,
the assignment get permuted accordingly when the items are permuted.
The unique mechanism works as follows. First, agents are sorted in a fixed 
order, and then the first agent chooses his favorite item, the next agent 
chooses his favorite item among remaining items, etc. When the fixed order 
is picked uniformly among all possible orderings, the resulted mechanism 
is called {\em Random Serial Dictatorship (RSD)}.

Besides the truthfulness, an important issue left is to understand
the efficiency of mechanisms in one-sided matching problems. The efficiency 
of a mechanism is defined as the social welfare of the assignment 
the mechanism returns. 
Zhou~\cite{zhou1990conjecture} confirmed Gale's conjecture by showing
that there is no \textit{symmetric, Pareto optimal} and \textit{truthful}
mechanism for general preferences. A mechanism is symmetric if agents are 
treated equally if they report the same preferences. A mechanism is Pareto
optimal if the mechanism never outputs an assignment that the social welfare 
could be improved without hurting any agent. It is well-known that RSD 
is truthful and ex post efficient, i.e., it never outputs Pareto dominated 
outcomes.

We observe that there is few work that study the
efficiency of RSD. The main reason is that its average
social welfare could be even $O\br n$ far away from the optimal
social welfare if the preferences of agents for items
are unrestricted. It happens when assigning a particular item to
a particular agent contributes most of the optimal social welfare. However,
in RSD it is possible that the agent only gets that item with a
probability of $1/n$. In this paper, we circumvent this problem by
considering smaller but still rich domains of preferences. %
The first type of preferences we consider is \textit{dichotomous
preferences}, where agents have binary preferences over items. We shall 
call this setting simply {\em dichotomous}. Dichotomous preferences are 
fairly natural in assignment problems. For example, professors indicate 
the courses they like or dislike to teach, or workers choose the working 
shifts they want. The goal here is to design good mechanisms to 
assign courses/shifts to professors/workers. One can model these 
problems with bipartite graphs: workers on one side, shifts on the other,
an edge indicates whether a worker wants to participate in a particular shift.
Then one can find a maximum matching in the graph to optimize the total 
value of the assignment. It is shown in~\cite{dughmi2010truthful} that 
with some careful tie-breaking rule, finding a maximum matching yields 
a truthful mechanism.
However, such mechanisms fail to capture the symmetry. 
To make this approach symmetric, one could find all maximum matchings 
and randomly choose one. Note that it implies that Zhou's impossibility 
result does not pertains to dichotomous preferences. However, since 
finding all maximum matchings in bipartite graphs is $\#P$-complete, 
we conjecture that it is computationally infeasible to design truthful 
and symmetric mechanisms that obtain optimal welfare. Therefore, we 
turn our attention to investigate how well mechanisms can approximate
the maximum social welfare. By the connection to the online bipartite
matching problem~\cite{KarpVV90,mahdian2011online}, we get the following
result: 

\begin{result}
In dichotomous setting there exists a truthful and symmetric mechanism
that is a $0.69$-approximation to the maximum social welfare.
\end{result}
Due to the limitation of such mechanisms, next we show that RSD also
obtains a constant approximation for dichotomous preferences. 

\begin{result}
Random Serial Dictatorship in dichotomous setting returns an assignment 
in which the expected social welfare is a 3-approximation of the maximum 
social welfare.
\end{result}

The second type of preferences we consider is \textit{normalized
von Neumann-Morgenstern preferences}, where the value of agent $i$
for item $j$ lies in $[0,1]$. We shall call this setting simply {\em normalized}.
In this setting our result gives asymptotically tight description of the social welfare achieved by RSD.

\begin{result}
In normalized setting with $n$ agents and $n$ items, Random Serial Dictatorship returns a matching which expected social welfare is at least $\frac{1}{e}\frac{\nu(\opt)^2}{n}$, where $\nu(\opt)$ is the maximum social welfare.
\end{result}
This result implies that RSD achieves an $\sqrt{e\cdot n}$-approximation
of the optimal social welfare in unit-range preferences, i.e., when
$\max_i v_a(i)=1$, $\min_i v_a(i)=0$.
Recently~\cite{FilosRatsikas} presented an $O\br{\sqrt{n}}$-approximation for RSD in unit-range setting.

Finally, we complement the above result with the following upper-bound.

\begin{result}\label{res:hardness}
Given $n$, for any $k=1,\ldots,n$ and for any $\epsilon > 0$
there exist an instance of one-sided matching problem with normalized von Neumann-Morgenstern           
preferences where $\sw{\opt}=k$ and no truthful mechanism can achieve expected social
welfare better than $\frac{k^2}{n}+\epsilon$, where $k$ is the optimal social welfare.

\end{result}

\subsection{Related work}
Here we only mention the most relevant work on one-sided matching problems.
For more details, we refer the reader to surveys~\cite{roth1992two,sonmez2011matching}.
One-sided matching problems modeled in~\cite{hylland1979efficient} gave a market-like procedure to produce efficient assignments.
There, the procedure is Pareto optimal but not truthful.
Gale and Shapley~\cite{gale1962college} considered a similar problem, the marriage problem, but they turned attention to the incentive issues on whether agents would or would not reveal their private preferences.
In~\cite{gale1987lotteries} authors were asking about existence of good mechanisms when preferences are also considered. Zhou~\cite{zhou1990conjecture} answered this question by showing that there is no \textit{symmetric, Pareto optimal} and \textit{truthful} mechanism.
Between ex-ante Pareto optimality and ex-post Pareto optimality, Bogomolnaia and Moulin~\cite{bogomolnaia2001new} introduced a new concept called {\em ordinal 
efficiency}.
They gave a {\em probabilistic serial} mechanism that always returns ordinal efficient assignments.
However, the probabilistic serial mechanism is not truthful.
Bhalgat et al.~\cite{bhalgat2011social} studied the efficiency of RSD in a more restricted setting than ours,
where agents have values of $\frac{n-j+1}{n}$ for their $j$th favorite item. 
Chakrabarty and Swamy~\cite{chakrabarty2014welfare} introduced the notion of rank approximation to measure the social welfare under ordinal preferences. One-sided matching problems with dichotomous preferences were studied by Bogomolnaia and Moulin~\cite{bogomolnaia2004random}. They used the Gallai-Edmonds decomposition of bipartite graphs to characterize the (most) efficient assignments.The most related work to ours is that Filos-Ratsikas et al.~\cite{FilosRatsikas} independently gave the similar approximation ratio of RSD under unit-range preferences while our results applies to more general settings.

Cardinal preferences enable agents to explicitly express how much they prefer each item, while this can not be done in ordinal preferences.
The space of cardinal preferences could be shown to be the same as the space of von Neumann-Morgenstern preferences.
In addition, the normalization of preferences is a standard procedure, see~\cite{kalai1977aggregation}.
Besides the literature of operational research and decision theory, normalized von Neumann-Morgenstern preferences
are widely used to model individual behavior in game-theoretical settings.

\section{Preliminaries}

\paragraph{The model}

We model one-sided matching problems as bipartite graphs. In a bipartite
graph, its left side is a set $A$ of agents and its right side are
a set $I$ of indivisible items. We assume $\size A=\size I=n$ and
each agent is matched to exactly one item. For each agent $a\in A$
and each item $i\in I$, there is an edge $\br{a,i}$ representing
a possible allocation of item $i$ to agent $a$. The preference of
agent $a$ for item $i$ is denoted by $v_{a}\br i$, which is the
value that agent $a$ has for item $i$. 
We consider two different types of preferences, {\em dichotomous
preferences} and {\em normalized von Neumann-Morgenstern preferences}.
In dichotomous preferences, it holds that $v_{a}\br i\in\{0,1\}$,
while in normalized von Neumann-Morgenstern preferences, it holds
that $v_{a}\br i\in\brq{0,1}$. In dichotomous case we shall say shortly
that agent $a$ {\em 1-values} item $i$, if $v_{a}\br i=1$, instead
of clunky ``agent $a$ has value 1 for item $i$''; the same with
value 0.

We say $v_{a}\br{\cdot}$ is the preference profile of agent $a$.
Denote by $\mathcal{V}$ the set of all possible preference profiles
of a single agent: for dichotomous preferences $\mathcal{V}=\set{0,1}^{I}$,
for normalized von Neumann-Morgenstern preferences $\mathcal{V}=\brq{0,1}^{I}$.
Preference profiles of all agents are denoted by $v_{A}=\br{v_{a}}_{a\in A}\in\mathcal{V}^{A}$;
by $v_{-a}=\br{v_{a'}}_{a'\in A\setminus a}$ we denote all profiles
except of agent $a$'s. By $\br{v'_{a},v_{-a}}$ we denote agents'
preferences with $a$'s preference changed from $v_{a}$ to $v'_{a}$;
if $\br{v'_{a},v_{-a}}$ is an argument of a function, then we skip
writing double brackets. 
Consider a set of items $I'\subseteq I$ and suppose that agent $a$
values items $i_{1},\ldots,i_{k}\in I'$ equally and more than any
other item in $I'$. We say that items $i_{1},\ldots,i_{k}$ are \emph{favorite}
items of agent $a$ in $I'$.

We call matrix $p_{A}=\br{p_{a}}_{a\in A}$, where $p_{a}=\br{p_{a}^{i}}_{i\in I}$,
a feasible {\em matching} if the following conditions hold: 1) for any $a\in A$ and $i\in I$, $p_{a}^{i}\in\set{0,1}$; 2) for any $a\in A$, $\sum_{i\in I}p_{a}^{i}=1$; 3) for any $i\in I$, $\sum_{a\in A}p_{a}^{i}=1$.
Given a feasible matching $p_{A}$, we say item $i$ is matched to
agent $a$ if $p_{a}^{i}=1$. Thus, the value of agent $a$ for the matching
$p_{A}$ is given by $v_{a}\cdot p_{a}=\sum_{i\in I}v_{a}\br ip_{a}^{i}$,
where $\cdot$ is an operator of the vector product. The social welfare
of the matching $p_{A}$ is given by $\sw{p_{A}}=\sum_{a\in A}v_{a}\cdot p_{a}$.

From each agent $a\in A$ mechanism $\mathcal{M}$ collects declarations
$d_{a}\in{\cal V}$ about his preference profile --- we overload notations
here a bit, since vector $d_{a}$ does not always have to be declared
completely, i.e., when some of the items are already matched, then
the mechanism does not ask $a$ about values for these items. Of course,
the connection between true valuations $v_{a}\in{\cal V}$ and declarations
$d_{a}\in{\cal V}$, which ${\cal M}$ collects, depends heavily on
the mechanism ${\cal M}$ itself. Mechanism ${\cal M}$ maps agents
declarations $d_{A}$ to a feasible matching ${\cal M}_{A}\br{d_{A}}$
(i.e.,~the $p_{A}$ matrix); ${\cal M}_{a}\br{v_{A}}$ denotes the
allocation to agent $a$ (i.e.,~the $p_{a}$ vector). Mechanism $\mathcal{M}$
might be randomized, and then matching ${\cal M}_{A}\br{d_{A}}$ is
a random matrix, and allocation ${\cal M}_{a}\br{v_{A}}$ is a random
vector as well. In this case, $\ex{\sw{{\cal M}_{A}\br{d_{A}}}}$
is the expected social welfare of mechanism ${\cal M}_{A}$, but since
all of the mechanisms we analyze are randomized, we shall call it
just social welfare.

We measure the performance of the mechanism by comparing the social
welfare it produces with the optimal social welfare $\sw{\opt\br{v_{A}}}$,
where $\opt\br{v_{A}}$ denotes a matching that maximizes the social
welfare when preferences are given by $v_{A}$. Note that $\opt\br{v_{A}}$
can be seen as a maximum weight matching in the graph $G=\br{A\cup I,A\times I}$
where weight of edge $\br{a,i}$ is equal to $v_{a}\br i$. For simplicity
however, throughout the paper we shall just write $\opt$, instead
of $\opt\br{v_{A}}$.

A mechanism $\mathcal{M}$ is {\em truthful}, if for every $a\in A$,
every $v_{A}\in\mathcal{V}^{A}$ and every $v'_{a}\in\mathcal{V}$,
it holds that (even when the mechanism is randomized)
\[
v_{a}\cdot\mathcal{M}_{a}\br{v_{A}}\geq v_{a}\cdot\mathcal{M}_{a}\br{v'_{a},v_{-a}}.
\]
A mechanism $\mathcal{M}$ is {\em symmetric} if for every $a_{1},a_{2}\in A$,
every $d_{A}\in\mathcal{V}^{A}$ such that $d_{a_{1}}=d_{a_{2}}$,
it holds that $\ex{{\cal M}_{a_{1}}\br{d_{A}}}=\ex{{\cal M}_{a_{2}}\br{d_{A}}}$,
i.e.,~agents with identical declarations have the same (expected) value for the
allocation.

\paragraph{RSD and iterative analysis}

Now let us give the formal description of the Random Serial Dictatorship
(RSD) mechanism. RSD first picks an ordering of agents uniformly at
random and then asks agents to choose sequentially with respect to
the order. We assume that agents are rational, i.e.,~they will always
choose the best items among the unmatched items. Ties are broken randomly,
i.e.,~when agent $a$ is asked in RSD and his favorite items are $i_{1}$
and $i_{2}$ among unmatched items, agent $a$ will chose items $i_{1}$
and $i_{2}$ with an equal probability. This is an important assumption
for the analysis of RSD with dichotomous preferences. If we would
like to analyze RSD when agent would always deterministically choose
among the best items, then the competitive ratio guarantees and lower
bounds from von Neumann-Morgenstern preferences would apply.

Let us observe a property of RSD that is important for our analysis.
Instead of thinking that a random ordering is fixed before any agent
is considered sequentially, we can think that RSD chooses an agent
randomly from remaining agents in each step. It is easy to see that
agents are considered in the same random order in both cases.

RSD is iterative in nature, and so is the analysis. Let
us index its time-steps by $t$, which ranges from $0$ to $n$. $t=0$
indicates the moment after sorting the agents, but before asking first
agent to choose. Let $\rsd^{t}$ represent the (partial) matching
constructed by RSD after first $t$ steps. Then $\sw{\rsd^{t}}$ represents
the social welfare obtained after first $t$ steps; in particular $\sw{\rsd^{0}}=0$.
As RSD is being executed, the set of unmatched agents and the set
of available items are gradually decreasing. Let $A^{t}$
and $I^{t}$ be the set of unmatched agents and the set of available
items after step $t$. For example, $A^{0}=A$ and $I^{0}=I$. As
the sets $A^{t}$ and $I^{t}$ are being modified, we also keep track
of the way in which $\sw{\opt}$ is being changed (recall that $\opt$
denotes a matching that maximizes the welfare). More precisely, we
start with $\sw{\opt^{0}}=\sw{\opt}$. Suppose that at step $t$,
RSD asks agent $a$ to choose and then $a$ picks item $i$, then
$\sw{\rsd^{t}}=\sw{\rsd^{t-1}}+v_{a}\br i$. We remove $a$ from $A^{t-1}$
and $i$ from $I^{t-1}$, e.g., $A^{t}=A^{t-1}-\set a$ and $I^{t}=I^{t-1}-\set i$.
In addition, we also remove welfare contributed by $a$ and $i$ from
$\sw{\opt^{t-1}}$. Certainly, when $t=n$, then $\sw{\opt^{n}}=0$,
while $\sw{\rsd^{n}}$ is the social welfare obtained by RSD.

Sequence $\brw{\sw{\rsd^{t}}}{}_{t\geq0}$, which represents the increasing
welfare of RSD, is a random process. Moreover, $\ex{\sw{\rsd^{n}}}$
represents the expected social welfare returned by RSD. The sequence
$\brw{\sw{\opt^{t}}}{}_{t\geq0}$, which represents how the optimal
social welfare is affected by the random choices within RSD, is a
random process as well. Therefore, we want to describe a relation
between $\ex{\sw{\rsd^{n}}}$ and $\sw{\opt^{0}}$, and to do so we
deploy theory of martingales.

\paragraph{Martingales}

Below we only introduce notions and properties that we use later in
the paper. For a thorough treatment of martingale theory see~\cite{probwithmartin}.

\begin{defn}\emph{ Consider a random process $\left(X^{t}\right)_{t=0}^{n}$.
Suppose we observe first $k$ steps of the process, and let ${\cal H}^{k}$
denote the information we have acquired in steps $0,1,\ldots,k$.
Expected value of $X^{k+1}$, conditioned on the information we have
from steps $0$ to $k$, is formally presented as $\excond{X^{k+1}}{{\cal H}^{k}}$.
If for any $k=0,\ldots,n-1$, we have $\excond{X^{k+1}}{\mathcal{H}^{k}}=X^{k}$,
then the process is called a }martingale.\end{defn} In other words,
the process does not change on expectation in one step. We shall also
consider a \emph{sub-martingale} $\left(X^{t}\right)_{t=0}^{n}$ which
satisfies $\excond{X^{k+1}}{H^{k}}\geq X^{k}$ instead of equality
in the above definition.\begin{customTheorem}{}[Doob's Stopping Theorem]
Let $\left(X^{t}\right)_{t=0}^{n}$ be a martingale,
respectively sub-martingale. For any $k=0,1,\ldots,n$ it holds that
$\ex{X^{k}}=\ex{X^{0}}$, respectively $\ex{X^{k}}\geq\ex{X^{0}}$.
\end{customTheorem} The above is not the Doob's theorem in its full generality,
but rather the simplest variant that still holds in our setting.

\section{Dichotomous preferences and online bipartite matching}

\label{sec:dichKVV}
In this section, we establish a connection between one-sided matching
with dichotomous preferences and online bipartite matching. A similar
connection was also presented in~\cite{bhalgat2011social}.

Consider a variant of online bipartite matching. We are given a bipartite
graph $G=(A\cup B,E)$, where one side $A$ of the graph is given,
while vertices from other side $B$ and edges between $A$ and $B$
are unknown. Suppose that vertices from $B$ arrive one by one, and
upon the arrival of vertex $b\in B$, all edges adjacent to $b$ are
revealed. On vertices of $A$ there is an ordering $\sigma$ given by
a random permutation. Consider RANKING algorithm that upon arrival
of vertex $b \in B$ it matches $b$ to the unmatched neighbor in $a\in A$
with the highest ranking $\sigma\br{a}$. In their seminal paper,
Karp et al\@.~{\cite{KarpVV90}} have proven that this algorithm
constructs a matching of expected size at least $\br{1-\frac{1}{e}}OPT$,
where $OPT$ is the offline optimum, and the bound holds even if the vertices
of $B$ arrive in an adversarial order. Furthermore, Mahdian and
Yan~\cite{mahdian2011online} have shown that the performance of RANKING
algorithm is even better when the order of vertices in $B$ is also given
by a random permutation:
\begin{customTheorem}{}
Given that the vertices in $B$ arrive uniformly at random and the order of
vertices in $A$ is random, RANKING algorithm constructs a matching of
expected size at least $0.69 \cdot OPT$, where $OPT$ is the offline optimum.
\end{customTheorem}

Now let us see consider the following mechanisms for one-sided matching with
dichotomous preferences. Given the agents and items, mechanism
RSD* generates a random ordering on agents and a random ranking on items. RSD* considers agents one by one according to the random ordering. Suppose that agent $a$ is considered at step $\tau$ and let $d_a(\cdot)$ be the preference reported by agent $a$. Denote by $I^\tau$ the
set of items yet unmatched at step $\tau$. If agent $a$ 1-values any unmatched item,
RSD* assigns agent $a$ an item with the highest rank among all remaining items.
Otherwise, RSD*  assigns nothing to agent $a$. Finally, RSD* matches any
unmatched items to  unmatched agents. Truthfulness of RSD* follows from the
observation that $\tau$ as well as $I^\tau$ are independent of $a$'s declaration $d_a$.
More precisely, the moment $\tau$ is given only by a random permutation of agents,
while set $I^\tau$ depends on the permutation of agents and declarations $d_{a'}$ of
agents $a'$ that came before $a$. Therefore, if $a$ declares $d_a\br{i}=1$ for item $i$ such that
$v_a\br{i}=0$, then he can only increase the probability that at moment $\tau$
he is matched to a 0-valued item. Analogically, if $a$ declares $d_a\br{i}=0$
for item $i$ such that $v_a\br{i}=1$, then he can only decrease the probability
that at moment $\tau$ he is matched to a 1-valued item. Suppose now that agent $a$ has
0 value for all items in $I^\tau$. In this case agent gains nothing regardless
of what his declarations are. An agent that reports truthfully in this case, we call
\emph{non-adversarial}.
Since RSD* is guided by two random permutations, the symmetry of
the mechanism is clear.

\begin{algorithm}
\label{alg:kvvrandom}
Let random permutation $\sigma:\set{1,...,n}\mapsto\set{1,...,n}$ be the ranking of items\;
For each agent $a\in A$ in
random order:\\
\qquad{}ask agent $a$ about his preference profile $d_{a}\in{\cal {V}}=\set{0,1}^{I}$\;
\qquad{}if there is no unmatched item $i$ such that $d_{a}\br{i}=1$, then discard agent $a$\;
\qquad{}otherwise, assign $a$ to unmatched item $i$ that has the highest rank $\sigma\br{i}$\;
Match any unmatched items to unmatched
agents anyhow. \caption{$\mbox{RSD*}(A,I)$}
\end{algorithm}

\begin{thm}
Assuming that agents are non-adversarial, RSD* is a truthful and
symmetric mechanism that achieves $0.69$-approximation to the maximum social
welfare in one-sided matching problems with dichotomous preferences.
\end{thm}

One can imagine that sometimes an agent can be adversarial,
and he would not admit that he does not value any of the remaining items.
To address this issue, in the next section
we present an analysis of RSD mechanism in which every agent can be adversarial.

\section{Dichotomous preferences and RSD}

\label{sec:dichRSD} \begin{thm} Random Serial Dictatorship always
returns an matching in which the expected welfare is at least $\frac{1}{3}\sw{\opt}$
in one-sided matching problems with dichotomous preferences.\end{thm}

\begin{proof}

Recall, $\opt$ is an optimal matching. Let $A^{t}$ be the set of
agents remaining after $t$ steps, let $I^{t}$ be the set of remaining
items, and $\opt^{t}\subset\opt$ is what remains from optimal solution
after $t$ steps of RSD. 
Also, $\rsd^{t}$ is the partial matching constructed by RSD after
$t$ steps, and $\sw{\rsd^{t}}$ be its welfare. For an agent $a$
let $\opt_{a}\in I$ be the item to which $a$ is matched in $\opt$.

Let $Y^{t}$ be the set of agents who are matched to an item in $\opt^{t}$
which they value 1, i.e., $\setst{a\in A^{t}}{v_{a}\br{\opt_{a}}=1}$.
Therefore $\size{Y^{t}}=\sw{\opt^{t}}$ for every $t$. It can happen
that at time $t$, an agent does not 1-value any of remaining items
$I^{t}$, even though he could have 1-valued some of the items in
$I^{0}$. Thus let $Z^{t}\subseteq A^{t}$ be the agents who 0-value
all items in $I^{t}$. Let us denote $y^{t}=\size{Y^{t}}$ and $z^{t}=\size{Z^{t}}$
for brevity.

Consider step $t+1$ of RSD, and assume we have all information available
after first $t$ steps, represented by ${\cal H}^{t}$. Let $a$ be
the agent who is to make his choice in this step, and let $i$ be
the item $a$ chooses. Agent $a$ does not belong to $Z^{t}$ with
probability $1-\frac{z^{t}}{n-t}$, and if this happens, then for
sure $v_{a}\br i=1$, which adds 1 to the welfare of RSD, i.e., $\sw{\rsd^{t+1}}=\sw{\rsd^{t}}+1$.
Hence $\excond{\sw{\rsd^{t+1}}}{{\cal H}^{t}}=\sw{\rsd^{t}}+1-\frac{z^{t}}{n-t}$.

Now let us analyze the expected decrease $\sw{\opt^{t}}-\sw{\opt^{t+1}}$.
Suppose that agent $a$ does not belong to $Z^{t}$, again with probability
$1-\frac{z^{t}}{n-t}$. Edge $\br{a,i}$ is adjacent to at most two
$1$-value edges in $\opt^{t}$, since $\opt^{t}$ is a feasible matching.
Thus when $a\notin Z^{t}$, then $\sw{\opt^{t}}-\sw{\opt^{t+1}}$
is at most 2 . Now suppose that agent $a$ belongs to $Z^{t}$, which
happens with probability $\frac{z^{t}}{n-t}$. Since $v_{a}\br i=0$,
then $a$ is \textbf{not} adjacent to any 1-value edge in $\opt^{t}$,
and $i$ may be adjacent to at most one such edge since agent $a$
choose an item randomly from unmatched items. Therefore, when $a\in Z^{t}$,
then $\sw{\opt^{t}}-\sw{\opt^{t+1}}$ is at most 1. Hence, together
with noting that $\frac{z^{t}}{n-t}+\frac{y^{t}}{n-t}\leq1$, we can
conclude that the expected decrease $\sw{\opt^{t}}-\sw{\opt^{t+1}}$
is: 
\[
\excond{\sw{\opt^{t}}-\sw{\opt^{t+1}}}{{\cal H}^{t}}\leq2\cdot\br{1-\frac{z^{t}}{n-t}}+\frac{z^{t}}{n-t}\cdot\frac{y^{t}}{n-t}\leq3\cdot\br{1-\frac{z^{t}}{n-t}}.
\]
 Since $\excond{\sw{\rsd^{t+1}}}{{\cal H}^{t}}=\sw{\rsd^{t}}+1-\frac{z^{t}}{n-t}$,
we get that 
\[
\excond{\sw{\opt^{t}}-\sw{\opt^{t+1}}}{{\cal H}^{t}}\leq3\cdot\br{1-\frac{z^{t}}{n-t}}=3\cdot\excond{\sw{\rsd^{t+1}}-\sw{\rsd^{t}}}{{\cal H}^{t}}.
\]

This means that sequence $\br{X^{t}}_{t=0}^{n}$, defined by $X^{0}=0$
and $X^{t+1}-X^{t}=3\cdot\br{\sw{\rsd^{t+1}}-\sw{\rsd^{t}}}-\br{\sw{\opt^{t}}-\sw{\opt^{t+1}}},$
satisfies $\excond{X^{t+1}}{{\cal H}^{t}}\geq X^{t}$, and therefore
is a sub-martingale. From Doob’s Stopping Theorem we get
that $\ex{X^{n}}\geq\ex{X^{0}}=0$, and hence
\begin{multline*}
0\leq\ex{X^{n}}=\ex{\sum_{t=1}^{n}X^{t}-X^{t-1}}=3\cdot\ex{\sum_{t=1}^{n}\sw{\rsd^{t}}-\sw{\rsd^{t-1}}}-\\
-\ex{\sum_{t=1}^{n}\sw{\opt^{t-1}}-\sw{\opt^{t}}}=3\cdot\ex{\sw{\rsd^{n}}}-\ex{\sw{\opt^{0}}},
\end{multline*}
 since $\rsd^{0}=\opt^{n}=\emptyset$. This allows us to conclude
that $3\cdot\ex{\sw{\rsd^{n}}}\geq\sw{\opt}$, which finishes the
proof. 
\end{proof}

Our analysis is simple, and most likely not tight --- approximation
ratio should be below 3. On the other hand, it is not very close
to 2, as there exist instances with dichotomous preferences in which
RSD gives expected outcome close to $\frac{1}{2.28}\cdot\sw{\opt}$.
One can see a resemblance between the following instance and the worst case
instance for algorithm RANDOM from Karp et al.~\cite{KarpVV90}.

\begin{fact}\emph{ Consider the following instance of a problem.
We have numbers $k$, $z$ and $n=z+k$, with $k$ even, and also
sets $A=\set{1,...,n}$, $I=\set{1,...,n}$. Define the valuations:
$v_{a}\br i=1$ if $a=i\in\set{1,\ldots,k}\mbox{ or }a\in\set{1,\ldots,\frac{k}{2}}\wedge i\in\set{\frac{k}{2},\ldots,k}$,
and 0 otherwise. The optimum solution in this case is obviously $k$.
Simulations indicate that for $k=10^{4}$ and $z=10^{7}$, the expected
performance of RSD is around $4378$ giving ratio of $\frac{10^{4}}{4378}\approx2.28$.
Taking different values of $k$ or $z$ did not significantly changed
the outcome of simulations. } \end{fact}

\section{Normalized von Neumann-Morgenstern preferences and RSD}\label{sec:normalRSD}
\begin{thm}
Random Serial Dictatorship always returns an assignment in which the
expected social welfare is at least $\frac{1}{e}\frac{\sw{\opt}^{2}}{n}$
in one-sided matching problems with normalized von Neumann- Morgenstern
preferences, where $\sw{\opt}$ is the maximum social welfare.
\end{thm}
\begin{proof}
As before, let $\opt$ be the optimal assignment, and $\opt^{t}\subseteq \opt$
be the subset of the optimal assignment that remains after $t$ steps of RSD.
Consider step $t+1$, and let ${\cal H}^{t}$ be all information
available after $t$ steps. We choose agent $a$ uniformly at random from
the remaining agents, and then $a$ chooses item $i$ that he prefers the most,
i.e., edge $\br{a,i}$ has the greatest value among edges
$\setst{\br{a,i}}{i\in I^{t}}$. The number of agents without an assigned
item is exactly $n-t$ after $t$ steps, and hence the probability of choosing
a particular agent is $\frac{1}{n-t}$.

Let $\opt\br{a}$ denote the item matched to agent $a$ in $\opt$.
Since agent $a$ has the largest value for item $i$ among remaining items,
it has to hold that $v_a\br{i}\geq v_a\br{\opt(a)}$. Therefore,
the expected welfare of RSD in step $t+1$ increases at least
\[
\sum_{a\in A^{t}}\frac{v_a\br{i}}{n-t}
\geq
\sum_{a\in A^{t}}\frac{v_a\br{\opt(a)}}{n-t}
=
\frac{\sw{\opt^{t}}}{n-t},
\]
and hence $\excond{\sw{\rsd^{t+1}}}{{\cal H}^{t}}\geq \sw{\rsd^{t}}+\frac{\sw{\opt^{t}}}{n-t}$.
Similar martingale-based reasoning as in Section~\ref{sec:dichRSD} yields that
$
\ex{\sw{\rsd^{n}}}\geq\ex{\sum_{t=0}^{n-1}\frac{\sw{\opt^{t}}}{n-t}}
$, so in the remaining part we give a lower bound on this sum.

When we remove agent $a$ and item $i$ in step $t+1$, what is the
average decrease $\sw{\opt^{t}}-\sw{\opt^{t+1}}$? Surely, we remove edge
$\br{a,\opt^{t}\br a}$ from $\opt^{t}$. However, item $i$ may be assigned
a different agent than $a$ in $\opt^{t}$, and the value of this assignment
can be arbitrary --- let us denote by $L^{t+1}\in\brq{0,1}$ the decrease
of $\opt^{t}$ caused by deleting the assignment of $i$. Therefore,
the average decrease at step $t+1$ is
$
\sw{\opt^{t}}-\excond{\sw{\opt^{t+1}}}{{\cal H}^{t}}
  =  \excond{L^{t+1}}{{\cal H}^{t}}+\frac{\sw{\opt^{t}}}{n-t},
$
so if we define sequence $\br{Y_{t}}_{t=1}^{n}$, where
\begin{equation}
Y^{t+1}=L^{t+1}+\frac{\sw{\opt^{t}}}{n-t}-\br{\sw{\opt^{t}}-\sw{\opt^{t+1}}},\label{eq:normalY}
\end{equation}
then $\excond{Y^{t+1}}{{\cal H}^{t}}=0$ for $t=0,1,\ldots,n-1$.
We define another sequence $\br{X^{t}}_{t=0}^{n}$ with $X^{0}=0$
and
$X^{t}=\sum_{i=1}^{t}Y^{i}$.

Equality $\excond{Y^{t+1}}{{\cal H}^{t}}=0$ implies $\excond{X^{t+1}}{{\cal H}^{t}}=X_{t}$,
which means that $\br{X^{t}}_{t=0}^{n}$ is a martingale, and from Doob's Stopping Theorem, we get that $0=\ex{X^{0}}=\ex{X^{n}}=  \ex{\sum_{t=1}^{n}Y^{t}}$.
Thus summing equality~\eqref{eq:normalY} for $t$ from 1 to $n-1$ and taking expectation yields that
\[
\ex{\sum_{t=0}^{n-1}\frac{\sw{\opt^{t}}}{n-t}}=\sw{\opt}-\ex{\sum_{t=1}^{n-1}L^{t}}.
\]
And since $\ex{\sum_{t=0}^{n-1}\frac{\sw{\opt^{t}}}{n-t}}$ is the outcome of RSD, we just need to upper-bound $\ex{\sum_{t=1}^{n-1}L^{t}}$ now.

Let us note that equality~\eqref{eq:normalY} can be transformed into
\[
\frac{Y^{t+1}}{n-t-1}=\frac{L^{t+1}}{n-t-1}-\br{\frac{\sw{\opt^{t}}}{n-t}-\frac{\sw{\opt^{t+1}}}{n-t-1}}
\]
for $t+1<n$. Since $\excond{Y^{t+1}}{{\cal H}^{t}}=0$, we have $\excond{\frac{Y^{t}}{n-t}}{{\cal H}^{t-1}}=0$
as well. Thus sequence $\br{Z^{t}}_{t=0}^{n-1}$ with $Z^{0}=0$ and
$Z^{t}=\sum_{i=1}^{t}\frac{Y^{i}}{n-i}$ is a martingale, and again
from Doob's Stopping Theorem we get that $0=\ex{Z^{0}} = \ex{Z^{n-1}} =\ex{\sum_{t=1}^{n-1}\frac{Y^{t}}{n-t}} $, which gives
\[
0 =  \ex{\sum_{t=1}^{n-1}\frac{Y^{t}}{n-t}}
  =  \ex{\sum_{t=1}^{n-1}\frac{L^{t}}{n-t}}-\ex{\sum_{t=1}^{n-1}\frac{\sw{\opt^{t-1}}}{n-t+1}-\frac{\sw{\opt^{t}}}{n-t}},
\]
and since the second sum telescopes we obtain that
\begin{equation}\label{eq:Lsum}
\ex{\sum_{t=1}^{n-1}\frac{L^{t}}{n-t}}
=\frac{\sw{\opt^0}}{n}-\ex{\sw{\opt^{n-1}}}
\leq \frac{\sw{\opt}}{n}.
\end{equation}
For any $L^{t}\in\brq{0,1}$ it holds that $\frac{L^{t}}{n-t}\geq\int_{t-L^{t}}^{t}\frac{dx}{n-x}$.
Moreover all intervals $\brq{t-L^{t},t}$ are disjoint, and they are
of total length of $\sum_{t=1}^{n-1}L^{t}$, hence
\[
\sum_{t=1}^{n-1}\frac{L^{t}}{n-t}  \geq  \sum_{t=1}^{n-1}\int_{t-L^{t}}^{t}\frac{dx}{n-x}\geq\int_{0}^{\sum_{t=1}^{n-1}L^{t}}\frac{dx}{n-x}
  = 
\ln\frac{n}{n-\sum_{t=1}^{n-1}L^{t}}.
\]
Function $x\mapsto\ln\frac{n}{n-x}$ is convex, so from Jensen's inequality and~\eqref{eq:Lsum} we get that
\[
\frac{\sw{\opt}}{n}\geq\ex{\sum_{t=1}^{n-1}\frac{L^{t}}{n-t}}\geq\ex{\ln\frac{n}{n-\sum_{t=1}^{n-1}L^{t}}}\geq\ln\frac{n}{n-\ex{\sum_{t=1}^{n-1}L^{t}}},
\]
which yields
$
n\br{1-e^{-\frac{\sw{\opt}}{n}}}\geq\ex{\sum_{t=1}^{n-1}L^{t}}.
$
We can now finish lowerbounding the outcome of RSD:
\begin{multline*}
\sw{\rsd}\geq\ex{\sum_{t=0}^{n-1}\frac{\sw{\opt^{t}}}{n-t}} =\sw{\opt}-\ex{\sum_{t=1}^{n-1}L^{t}}
\geq \sw{\opt}-n+n\cdot e^{-\frac{\sw{\opt}}{n}} \geq \frac{1}{e}\frac{\sw{\opt}^2}{n},
\end{multline*}
where the last inequality follows from $x-1+e^{-x}\geq\frac{1}{e}x^{2}$
for $x\in\brq{0,1}$.
\end{proof}

The theorem above can be easily applied to the case that agents' preferences are unit-range. 

\begin{cor}
\label{cor:unitrange}
When agents' preferences are unit-range, i.e., $\max_i v_a(i)=1$, $\min_i v_a(i)=0$, for $a\in A$, Random Serial Dictatorship $\sqrt{e\cdot n}$-approximates the maximum social welfare.
\end{cor}

\begin{proof}
In unit-range preferences each agent has value $1$ for at least one item,
 hence RSD gets exactly welfare $1$ in the first step. Therefore, 
$\sw{\rsd}\geq 1$ and it means that the approximation ratio is at least 
$\frac{1}{\sw{\opt}}$. Since we have shown that RSD achieves 
at least $\frac{1}{e}\frac{\sw{\opt}^{2}}{n}$, the approximation ratio 
of RSD is at least $\max\brw{ \frac{1}{\sw{\opt}}, \frac{\sw{\opt}}{e\cdot n}}\geq \frac{1}{\sqrt{e\cdot n}}$.
\end{proof}

On the hardness side, we can show that no truthful mechanism
can do significantly better.

\begin{thm}
\label{thm:hardness}
Given $n$, for any $k=1,\ldots,n$ and for any $\epsilon > 0$
there exist an instance of one-sided matching problem with normalized von Neumann-Morgenstern           
preferences where $\sw{\opt}=k$ and no truthful mechanism can achieve expected social
welfare better than $\frac{k^2}{n}+\epsilon$, where $k$ is the optimal social welfare.

\end{thm}

\begin{figure}[htbp]
\begin{center}
\subfloat{
\scalebox{0.5}{ \input{hardness.pstex_t}}
}
\caption{\label{fig:hardness}}
\end{center}
\end{figure}

Consider an instance presented in Figure~\ref{fig:hardness}. Agent $a_1$ has
value 1 for item $i_1$, and any other player $a_i$, $i=2,...,\bar{n}$, has value
$\varepsilon$ for item $i_1$, where $\epsilon$ is a small quantity. All agents have value $0$ for items $i_2,i_3,...,i_{\bar{n}}$.
Obviously assigning item $i_j$ to agent $a_j$ is an optimum assignment
and it has value $\sw{\opt}=1$.
Since we cannot distinguish between agents, we need to assign them item $i_1$
with the same probability --- this means that any truthful mechanism can not
achieve welfare better than $\frac{1}{\bar{n}}+\frac{\bar{n}-1}{\bar{n}}\varepsilon$.
This is made formal in the following Lemma.

\begin{lem}\label{lem:hardness}
There exists an instance (see Figure~\ref{fig:hardness}) such that $\sw{\opt}=1$ but any truthful mechanism cannot achieve an expected social welfare
better than $\frac{1}{\bar{n}}+\varepsilon$.
\end{lem}
\begin{proof}
Let us consider the first instance as follows
\[
v^1(a,i)=\begin{cases}
0 & \mbox{if }1\leq a\leq n,2\leq i\leq \bar{n}\\
\epsilon & \mbox{if }1\leq a\leq \bar{n},i=1
\end{cases}
\]
where $\epsilon$ is a small quantity. In this case, consider any mechanism, it is cleat that there exists an agent who obtains item $1$ with a probability at most $\frac{1}{\bar{n}}$. Without loss of generality, we assume that agent $1$ is such an agent. Now let us consider the second instance in Figure~\ref{fig:hardness}. 
\[
v^2(a,i)=\begin{cases}
0 & \mbox{if }1\leq a\leq n,2\leq i\leq \bar{n}\\
\epsilon & \mbox{if }2\leq a\leq \bar{n},i=1\\
1 & \mbox{if }a=1,i=1
\end{cases}
\]
The optimal social welfare is $1$ by assigning item $1$ to the first agent. It is also easy to see that any mechanism that achieves an approximation ratio better than $O(\frac{1}{\bar{n}})$ must  allocate item 1 to agent 1 with a probability larger than $\frac{1}{\bar{n}}$. It implies that, under any truthful mechanism with an approximation ratio better than $O(\frac{1}{\bar{n}})$, agent $1$ in the first instance could benefit by misreporting his values as in the second instance. This proves that no truthful mechanism could achieve an expected social welfare better than $\Omega(\frac{\sw{\opt}}{\bar{n}})$ in the second instance where $\sw{\opt}=1$. 
\end{proof}

Using $k$ copies of this instance in Figure~\ref{fig:hardness}, we can prove Theorem~\ref{thm:hardness}.

\emph{Proof of Theorem~\ref{thm:hardness}.}
For simplicity let us assume that $k$ divides $n$.
Consider now the following instance with $n$ agents and $n$ items. We divide
agents and items into $k$ chunks, each consisting of $\frac{n}{k}$
 agents and the same number of items. Each chunk looks exactly like the
instance from Figure~\ref{fig:hardness} where $\bar{n}=\frac{n}{k}$ and $\varepsilon = \epsilon/k$.
Agents have value 0 for items from different chunks. Therefore social welfare
of any mechanism is a sum of welfares in all chunks. From
Lemma~\ref{lem:hardness} we know that on each chunk any truthful
mechanism gets an expected social welfare of at most $\frac{k}{n}+\varepsilon$.
Since there are $k$ chunks, no truthful mechanism can get
an expected social welfare on the whole instance better than
$k\cdot \br{\frac{k}{n} +\varepsilon}= \frac{k^2}{n} +\epsilon$.
On the other hand, each of the $k$ chunks contributes 1 to
the optimal welfare, giving $\sw{\opt}=k$.
This concludes the proof.
\qed

\section{Open question}
As mentioned in the introduction, we can give the following truthful and symmetric mechanisms that outputs optimal social welfare.
The mechanism works as follows. First, collect agents preferences $d_a$ for all $a\in A$. Then consider graph $G=(A,I)$ with edge between every pair $a \in A$, $i \in I$ for which $d_a(i)=1$..
Next, find the all maximum matchings.
Finally, output a maximum matching uniformly at random.
\begin{cl}
The mechanism above is truthful and symmetric, and outputs optimal social welfare.
\end{cl}
\begin{proof}
The symmetry and optimality of the mechanism is easy to see since it outputs one of the maximum matching uniformly at random. The following shows that the mechanism is also truthful. 

Let $d_a$ be the declared preference profile of agent $a$, and let $d_{-a}$ be declarations of all agents but $a$. Consider item $i$ which $a$ values 0, and suppose $a$ declares $d_a\br{i}=0$.
Let $M_{A}$ be the number of all maximum matchings,
let $M^1_a$ be the number of matchings in which $a$ is assigned item he 1-values.
Therefore expected value of $a$'s assignment is $\frac{M^1_a}{M_{A}}$.
Suppose now that $a$ would declare $d_a\br{i}=1$ instead.
There are two situations: with this change the size of maximum matching has increased by one, or remained the same.
If the size increased by $1$, then it means that right now all matchings use edge $\br{a,i}$, and in this situation $a$ is always assigned
item $i$, which he 0-values.
Hence, he does not have incentive to misreport in this case.
If the size of maximum matching remained the same, then the total number of matchings
could only increase (or remain the same) and now is equal to $M_A+M_{\br{a,i}}$.
However, the number of matchings in which $a$ is assigned 1-valued item, remains the same: $M^1_a$.
Therefore, after misreporting value of $i$, agent $a$ has probability of receiving 1-valued item equal to $\frac{M^1_a}{M_A+M_{\br{a,i}}}\leq \frac{M^1_a}{M_{A}}$,
Hence, $a$ does not have incentive to misreport in this case either.

Consider the other situation. Let $i$ be an item which $a$ values 1, and suppose $a$ declares $d_a\br{i}=1$.
As before, let $M_A$ be the total number of matchings,
let $M^1_a$, be the number of matchings in which $a$ is assigned item he 1-values.
Now the probability that $a$ is matched to 1-valued item is equal to $\frac{M^1_a}{M_A}$.
Suppose that $a$ declares $d_a\br{i}=0$.
After $a$ has changed his declaration, we have two possibilities: size of maximum matching has decreased by one, or remained the same.
If the size has decreased by one, then it means that $a$ is not assigned anymore to any item, so he gets value of 0 in this case, and hence he does not have any incentive to lie.
If the size has remained the same, then the total number of matchings is now equal to $M_A-M_{\br{a,i}}$.
But the number of matchings in which $a$ was assigned 1-value item, decreases by the same amount, i.e., $M^1_a-M_{\br{a,i}}$ is now
the number of matchings from which $a$ benefits value 1.
Therefore, the probability of receiving 1-valued item is now equal to $\frac{M^1_a-M_{\br{a,i}}}{M_A-M_{\br{a,i}}}$, and
\[
 \frac{M^1_a-M_{\br{a,i}}}{M_A-M_{\br{a,i}}} \leq \frac{M^1_a}{M_A},
\]
for any $M_{\br{a,i}}$. Hence $a$ does not have incentive to misreport in this case either.
\end{proof}

Unfortunately, such a mechanism is not feasible when computational efficiency is required.
The problem is that it is $\#P$-complete to count all maximum matchings.
Therefore, we suspect that any truthful, symmetric and optimal mechanism
would be somehow connected with an algorithm for counting all maximum matchings.
And because of that, we conjecture that such mechanism should be $\#P$-complete as well.

\end{document}